\documentclass[12pt]{amsart}
\usepackage{eurosym}
\usepackage{amssymb}
\usepackage{amsfonts}
\usepackage{graphicx,color}

\newcommand{\abs}[1]{\ensuremath{\left| #1 \right| }}

\theoremstyle{plain}
\newtheorem{corollary}{Corollary}[section]
\newtheorem{lemma}[corollary]{Lemma}

\newtheorem{remark}[corollary]{Remark}
\newtheorem{theorem}[corollary]{Theorem}
\numberwithin{equation}{section}

\begin{document}
\title{MSE estimates for multitaper spectral estimation and off-grid compressive
sensing}
\author{Lu\'{\i}s Daniel Abreu}
\email{labreu@kfs.oeaw.ac.at}
\author{Jos\'{e} Luis Romero}
\email{jlromero@kfs.oeaw.ac.at}
\address{Acoustics Research Institute, Austrian Academy of Science,
Wohllebengasse 12-14 A-1040, Vienna, Austria}
\subjclass{}
\keywords{}
\thanks{L. D. A. was supported by the Austrian Science Fund (FWF)
START-project FLAME ("Frames and Linear Operators for Acoustical Modeling
and Parameter Estimation") Y 551-N13. J. L. R. gratefully acknowledges support
from the Austrian Science Fund (FWF): P 29462 - N35
and from a Marie Curie fellowship, within the 7th. European Community Framework program, under grant PIIF-GA-2012-327063.
The authors also acknowledge support from the
WWTF grant INSIGHT (MA16-053).}

\begin{abstract}
We obtain estimates for the Mean Squared Error (MSE) for the multitaper spectral estimator and
certain compressive acquisition methods for multi-band signals.
We confirm a fact discovered by Thomson
{[Spectrum estimation and harmonic analysis, Proc. IEEE, 1982]}: assuming bandwidth $W$
and $N$ time domain observations, the average of the square of the
first $K=\left\lfloor 2NW\right\rfloor$ Slepian functions
approaches, as $K$ grows, an ideal band-pass kernel for the interval
$\left[ -W,W\right]$. We provide an analytic proof of this fact and
measure the corresponding rate of convergence in the $L^{1}$ norm. This
validates a heuristic approximation used to control the MSE of the multitaper
estimator. The estimates have also consequences for the method of compressive acquisition of
multi-band signals introduced by Davenport and Wakin, giving MSE approximation bounds for the
dictionary formed by modulation of the critical number of prolates.
\end{abstract}

\maketitle

\section{Introduction}
The \emph{discrete prolate spheroidal sequences} (DPSS),
introduced by Slepian in \cite{SlepianV}, play a fundamental role both in
Thomson's multitaper method for spectral estimation of stationary processes \cite{Thomson},
and in the method proposed by Davenport and Wakin for compressive acquisition
of multi-band analog signals in the presence of off-grid frequencies \cite{DW}. (See also
\cite{HogLak,DB,adha16,FST}.)

In this article we provide Mean Squared Error (MSE) estimates for both
methods based on the following observation: while the description of each individual discrete
Slepian function can be very subtle, the \emph{aggregated behavior} of the critical
number of solutions displays a simple profile, that can be quantified (see
Figures \ref{fig_1} and \ref{fig_2}).

Our bounds for the expectation of Thomson's estimator validate
some heuristics from \cite{Thomson} and elaborate on more qualitative work \cite{bronez,liro08}.
Similar performance bounds were until now only available for modified versions of Thomson's
method, that replace the Slepian sequences with other tapers that have analytic
expressions \cite{risi95}.

\section{Thomson's multitaper method}
\label{sec_tom}

Let $I=\left[ -1/2,1/2\right] $. Any stationary, real, ergodic, zero-mean,
Gaussian stochastic process has a \emph{Cram\'{e}r spectral representation}
\begin{equation*}
x(t)=\int_{I}e^{2\pi i\xi t}dZ(\xi )\text{,}
\end{equation*}
and the \emph{spectrum} $S(\xi )$, defined as
\begin{equation*}
S(\xi )d\xi =\mathbb{E}\{\left\vert dZ(\xi )\right\vert ^{2}\}\text{,}
\end{equation*}
and often called the \emph{power spectral density} of the process, yields
the periodic components of $x(t)$. The goal of spectral estimation is to
solve the highly underdetermined problem of \emph{estimating }$S(\xi )$
\emph{\ from a sample of }$N$\emph{\ contiguous observations }$
x(0),...,x(N-1)$. Embryonic approaches to the problem \cite{stokes, schuster}
used the so called \emph{periodogram}:
\begin{equation}
\widehat{S}(\xi )=\frac{1}{N}\left\vert \sum_{t=0}^{N-1}x(t)e^{-2\pi i\xi
t}\right\vert ^{2},
\end{equation}
whose analysis has influenced harmonic analysts since Norbert Wiener (see
\cite{Benedetto}). The periodogram can also be weighted with a data window
$\left\{ D_{t}\right\} _{t=0}^{N-1}$, usually called a\emph{\ taper}, giving
the estimator:
\begin{equation}
\widehat{S}_{D}(\xi )=\left\vert \sum_{t=0}^{N-1}x(t)D_{t}e^{-2\pi i\xi
t}\right\vert ^{2}\text{.}  \label{direct}
\end{equation}
The choice of the taper $\left\{ D_{t}\right\} _{t=0}^{N-1}$ can have a
significant effect on the resulting spectrum estimate $\widehat{S}_{D}$.
This is apparent by observing that its expectation is the convolution of the
\emph{true (nonobservable) spectrum} $S(\xi )$ with the \emph{spectral
window } $\left| \mathcal{F}D(\xi ) \right| ^2= \left|
\sum_{t=0}^{N-1}D_{t}e^{-2\pi i\xi t} \right| ^2$, i.e.,
\begin{equation}  \label{smooth}
\mathbb{E}\left\{ \widehat{S}_{D}(\xi )\right\} =S(\xi )\ast
\left| \mathcal{\ F}D(\xi ) \right| ^2.
\end{equation}
Thus, the bias of the tapered estimator, which is the difference $S(\xi) -
\mathbb{E}\{\widehat{S}_{D}(\xi )\}$, is determined by the smoothing effect
of $\left\{ D_t \right\}_{t=0}^{N-1}$ over the true spectrum. Ideally, the
function $\mathcal{F}D(\xi )$ should be concentrated on the interval $[-
\tfrac{1}{2N},\tfrac{1}{2N}]$, but the \emph{uncertainty principle of
Fourier analysis} precludes such perfect concentration (see e.g. \cite{HogLak}).

In \cite{Thomson}, Thomson used the DPSS basis to construct an algorithm that averages several
tapered estimates, whence
the name \emph{multitaper}. Thomson's multitaper method has been used in a variety of scientific
applications including climate analysis (see, for instance \cite{Science1997}, or \cite{PNAS2012}
for a local spherical approach), statistical signal analysis \cite{PW}, and it was used to
better understand the relation between atmospheric $CO_{2}$ and climate
change (see \cite[Section 1]{Thomson2}).

Today, Thomson's multitaper method remains an effective spectral estimation
method. It has recently found remarkable applications in
electroencephalography \cite{Neurosciences} and it is the preferred spectral
sensing procedure \cite{Cognitive} for the rapidly emerging field of
cognitive radio \cite{Haykin} \cite[Chapter 3]{HogLak}. In the next paragraph we provide an outline
of the method.

Thomson's method starts by selecting a target frequency smoothing band
$[-W,W]$ with $1/2N<W<1/2$, thus accepting a reduction in spectral resolution
by a factor of about $2NW$. The \emph{first step} consists of obtaining a
number $K=\left\lfloor 2NW\right\rfloor $ (the smallest integer not greater
than $2NW $) of estimates of the form \eqref{direct} by setting, for every
$k\in \{0,...,K-1\}$, $D_{t}=v_{t}^{(k)}(N,W)$, where the \emph{discrete
prolate spheroidal sequences} $v_{t}^{(k)}(N,W)$ are defined as the
solutions of the Toeplitz matrix eigenvalue equation
\begin{equation}
\label{eq_mee}
\sum_{n=0}^{N-1}\frac{\sin 2\pi W\left( t-n\right) }{\pi \left( t-n\right) }
v_{n}^{(k)}(N,W)=\lambda _{k}(N,W)v_{t}^{(k)}(N,W),
\end{equation}
and normalized by $||v^{(k)}||_2=1$.
The resulting tapered periodogram is then denoted by $\widehat{S}_{k}(\xi )$.
The \emph{second step}
consists of averaging. One uses the estimator
\begin{equation}
\widehat{S}_{(K)}(\xi )=\frac{1}{K}\sum_{k=0}^{K-1}\widehat{S}_{k}(\xi ),
\label{Thomson}
\end{equation}
which achieves a reduced variance (see \cite{Thomson} for an asymptotic
analysis of slowly varying spectra and \cite{WMP,liro08} for non-asymptotic
expressions).

To inspect the performance of the estimator $\widehat{S}_{(K)}(\xi )$ on the
spectral domain, let us consider the \emph{discrete prolate spheroidal
functions}, also known as \emph{Slepians}. They are the discrete Fourier
transforms of the sequences $v_{t}^{(k)}(N,W)$, denoted by $U_{k}(N,W;\xi )$,
and satisfy the
integral equation
\begin{equation}
\int_{-W}^{W}\mathbf{D}_{N}(\xi -\xi ^{\prime })U_{k}(N,W;\xi^{\prime })d\xi
^{\prime }=\lambda _{k}(N,W)U_{k}(N,W;\xi )\text{,}  \label{eq_prolates}
\end{equation}
where
\begin{equation}
\mathbf{D}_{N}(x)=\frac{\sin N\pi x}{\sin \pi x}  \label{Dirichlet}
\end{equation}
is the Dirichlet kernel. Observe that, according to (\ref{smooth}),
$\mathbb{E}\{\widehat{S}_{k}(\xi )\}$ is a smoothing average of the unobservable
spectrum by the kernel $\left\vert U_{k}(N,W;\xi )\right\vert ^{2}$. Recall
that the bias of each individual estimate in (\ref{Thomson}) is given by
\begin{equation}
Bias\left( \widehat{S}_{k}(\xi )\right) =\mathbb{E}\{\widehat{S}_{k}(\xi
)\}-S(\xi )=S(\xi )\ast\left\vert U_{k}(N,W;\xi )\right\vert ^{2}-S(\xi ).
\label{bias}
\end{equation}
The optimal concentration of the first prolate function on the interval
$[-W,W]$ leads to a low bias when $k=0$. But since the amount of energy of
$U_{k}(N,W;\xi )$ inside $[-W,W]$ decreases with $k$ (because the energy is
given by the eigenvalues in \eqref{eq_prolates} and they decrease from $1$
to just above $0$ as $k$ crosses the critical value $K$), the bias may increase for large
values of $k$. Let us inspect the \emph{averaged} estimator. Its expectation is
\begin{equation}
\mathbb{E}\{\widehat{S}_{(K)}(\xi )\}=\frac{1}{K}\sum_{k=0}^{K-1}\mathbb{E}
\{ \widehat{S}_{k}(\xi )\}=S(\xi ) \ast \frac{1}{K}\rho _{K}(N,W;\xi )\text{,
}  \label{Sk}
\end{equation}
where
\begin{equation}
\frac{1}{K}\rho _{K}(N,W;\xi )=\frac{1}{K}\sum_{k=0}^{K-1}\left\vert
U_{k}(N,W;\xi )\right\vert ^{2}.  \label{eq_inten}
\end{equation}
To explain the bias performance of the averaged estimator, Thomson \cite[Section 4]{Thomson}
observed that the \emph{spectral window} (\ref{eq_inten}) is very similar to a flat function
localized on $[-W,W]$ (see Figures \ref{fig_1} and \ref{fig_2}). This is an intriguing mathematical phenomenon.
Heuristically, it requires the functions in the sequence $\{\left\vert
U_{k}(N,W,\cdot )\right\vert ^{2}:k=0,\ldots, K-1\}$ to be organized inside
the interval $[-W,W]$ in a very particular way: \emph{each function tends to
fill in the empty energy spots left by the sum of the previous ones} - a
behavior reminiscent of the Pythagorean relation for pure frequencies: $
\sin ^{2}(t)+\cos ^{2}(t)=1$. More precisely, claiming that the spectral
window in Thomson's method approximates an ideal band-pass kernel means
that the two functions
\begin{equation}
\frac{1}{K}\rho _{K}(N,W,.)\text{ \ \ \ and \ \ \ \ }\frac{1}{2W}\mathbf{1}
_{[-W,W]}\text{,}  \label{two}
\end{equation}
approach each other as $K$ increases. This is indeed true and we provide an
analytic bound for the $L^{1}$-distance between the functions in \eqref{two}.

\begin{theorem}[Spectral leakage estimate]
\label{th_main} Let $N\geq 2$ be an integer, $W\in (-1/2,1/2)$ and set $
K:=\left\lfloor 2NW\right\rfloor $. Then
\begin{equation}
\left\Vert \frac{1}{K}\rho _{K}(N,W,\cdot )-\frac{1}{2W}\mathbf{1}
_{[-W,W]}\right\Vert _{L^{1}(I)}\lesssim \frac{\log N}{K}.  \label{L1}
\end{equation}
\end{theorem}

\begin{figure}[tbp]
\centering
\includegraphics[scale=0.45]{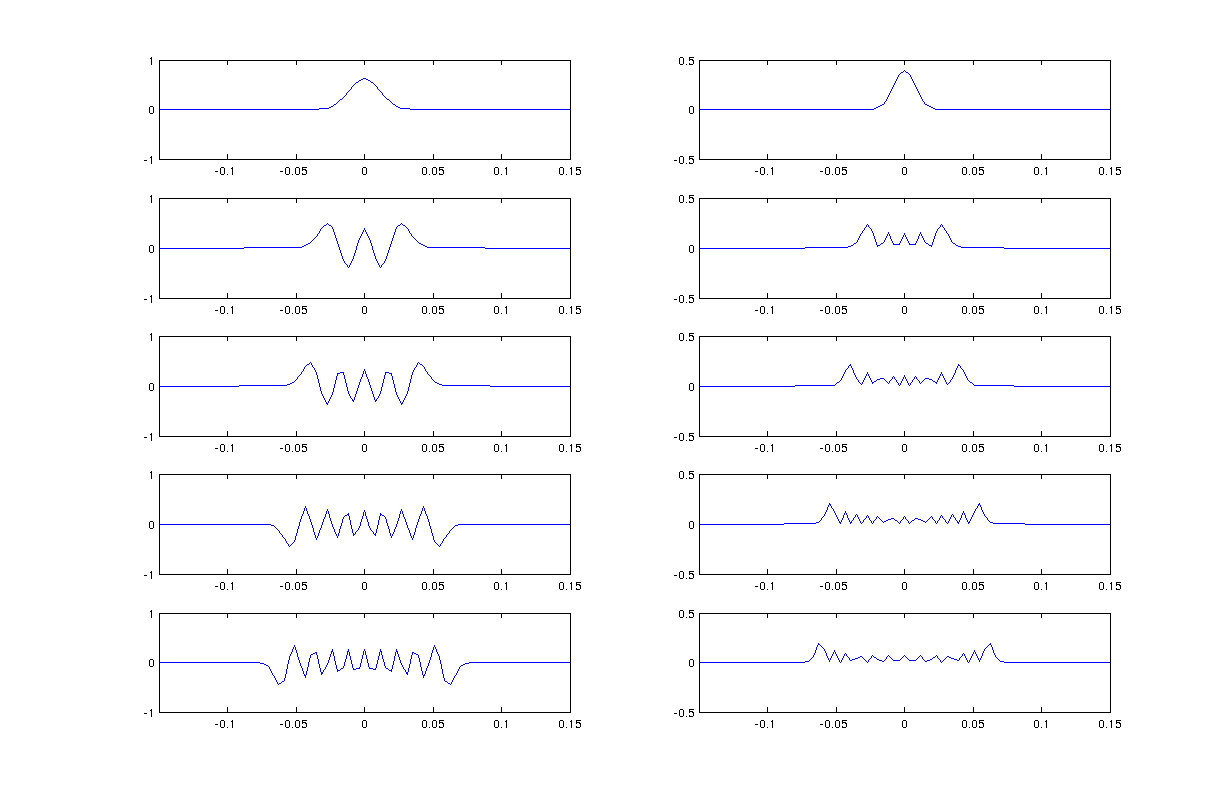}
\label{fig_prolates}
\caption{Slepians $U_k$ and their squares $\abs{U_k}^2$, for $N=256$ and $W=0.1$ and $k=1,5,9,19$.}
\label{fig_1}
\end{figure}

\begin{figure}[tbp]
\centering
\includegraphics[scale=0.40]{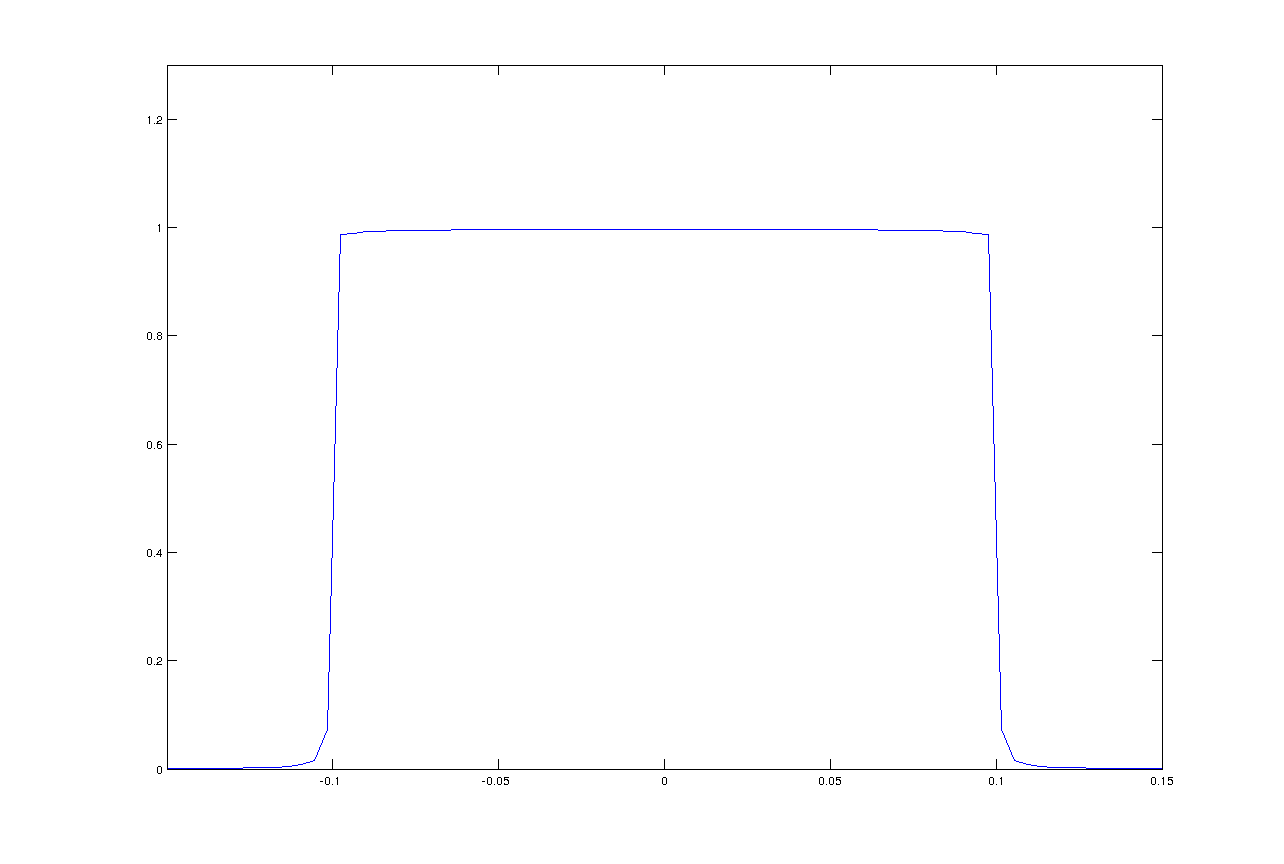}
\caption{Thomson's spectral window, with $N=256$ and $W=0.1$.}
\label{fig_2}
\end{figure}

The estimate \eqref{L1} is precisely what we need in order
to quantify Thomson's asymptotic analysis of the bias of the multitaper
estimator \cite[pg.. 1062]{Thomson} and validate the bias-variance
trade-off. Indeed, the $L^1$ deviation estimate allows one to
control the effect of replacing the spectral window by an ideal kernel in \eqref{Sk}.
This is explained in detail in Section \ref{sec_conclusion}.

As an alternative to \eqref{Thomson}, Thomson also considered a modified
method, where all Slepian sequences are used as tapers - instead of just the
first $K$ - but the corresponding tapered periodograms are weighted with the
eigenvalues. In this case, a similar analysis applies and the spectral
window is:
\begin{align}  \label{eq_inten_modif}
\tilde \rho _{K}(N,W,\xi) := \sum_{k=0}^{N-1} \lambda_k(N,W) \left\vert
U_{k}(N,W;\xi )\right\vert ^{2}.
\end{align}

We provide similar bounds for the modified method.

\begin{theorem}
\label{L1p} Let $N\geq 2$ be an integer and $W\in (-1/2,1/2)$. Then
\begin{equation}
\left\Vert \frac{1}{K}\tilde \rho _{K}(N,W,\cdot )-\frac{1}{2W}\mathbf{1}
_{[-W,W]}\right\Vert _{L^{1}(I)}\lesssim \frac{\log N}{K}.
\end{equation}
\end{theorem}

\section{Analysis of Thomson's spectral window}
\label{sec_proof}
Let $I:=[-1/2,1/2]$ and let us denote the exponentials
by $e_{\omega}(x):=e^{2\pi ix\omega }$. We will always let $N\geq 2$ be an integer and $W\in
(-1/2,1/2)$. For two non-negative functions $f,g$, the notation $f\lesssim g$ means that there
exists a constant $C>0$ such that $f\leq Cg$. (The constant $C$, of course, does not depend on the
parameters $N,W$.)

We normalize the Slepian functions by $\int_{I}\left\vert U_{k}(N,W;\xi
)\right\vert ^{2}\,d\xi =1$. We will need a
description of the profile of the eigenvalues in \eqref{eq_mee}. The following
lemma will be key in most of the estimates.

\begin{lemma}
\label{lemma_sum_eig} For $N \geq 2$, $W \in (-1/2,1/2)$ and $K:=\left\lfloor 2NW\right\rfloor$:
\begin{equation}  \label{eq_bound}
\left| 1 - \frac{1}{K}\sum_{k=0}^{K-1} \lambda_k(N,W) \right| \lesssim \frac{
\log N}{K}.
\end{equation}
\end{lemma}

We postpone the proof of Lemma \ref{lemma_sum_eig} to Section \ref{sec_tech}. The
quantity on the left-hand side of \eqref{eq_bound} has been studied in \cite{liro08} to
qualitatively analyze the performance of Thomson's method. Lemma \ref{lemma_sum_eig} refines the
analysis of \cite{liro08}, giving a concrete growth estimate. (See also the remarks after Theorem 5
in \cite{liro08}.)

\subsection{\textbf{Proof of Theorem }\protect\ref{th_main}}
We first estimate the narrow band error. Note that
\begin{equation*}
\rho _{K}(N,W;\xi )=\sum_{k=0}^{K-1}\left\vert U_{k}(N,W;\xi )\right\vert
^{2}\leq \sum_{k=0}^{N-1}\left\vert U_{k}(N,W;\xi )\right\vert
^{2}=D_{N}(0)=N\text{.}
\end{equation*}
Consequently, $\tfrac{1}{K}\rho _{K}(N,W;\xi )\leq \tfrac{N}{K}$ and, using \eqref{eq_eigen}, we can
estimate:
\begin{align*}
& \int_{-W}^{W}\left\vert \frac{1}{K}\rho _{K}(N,W;\xi )-\frac{1}{2W}\mathbf{1}_{[-W,W]}(\xi
)\right\vert \,d\xi  \\
& \qquad \leq \int_{-W}^{W}\left\vert \left( \frac{1}{2W}-\frac{N}{K}\right)
\mathbf{1}_{[-W,W]}(\xi )\right\vert \,d\xi +\int_{-W}^{W}\left\vert \frac{1}{K}\rho _{K}(N,W;\xi
)-\frac{N}{K}\mathbf{1}_{[-W,W]}(\xi )\right\vert
\,d\xi  \\
& \qquad =2W\left( \frac{N}{K}-\frac{1}{2W}\right)
+\frac{2NW}{K}-\frac{1}{K}\sum_{k=0}^{K-1}\int_{-W}^{W}\left\vert U_{k}(N,W;\xi )\right\vert
^{2}\,d\xi  \\
& \qquad \leq \frac{2}{K}+1-\frac{1}{K}\sum_{k=0}^{K-1}\lambda _{k}\lesssim
\frac{\log N}{K},
\end{align*}
thanks to Lemma \ref{lemma_sum_eig}. Now we estimate the broad brand
leakage:
\begin{align*}
& \int_{I\setminus \lbrack -W,W]}\left\vert \frac{1}{K}\rho _{K}(N,W;\xi
)-\frac{1}{2W}\mathbf{1}_{[-W,W]}(\xi )\right\vert \,d\xi =\int_{I\setminus
\lbrack -W,W]}\frac{1}{K}\rho _{K}(N,W;\xi )\,d\xi  \\
& \qquad =\frac{1}{K}\sum_{k=0}^{K-1}(1-\lambda _{k})=1-\frac{1}{K}
\sum_{k=0}^{K-1}\lambda _{k},
\end{align*}
so the conclusion follows invoking again Lemma \ref{lemma_sum_eig}.

\subsection{\textbf{Proof of Theorem }\protect\ref{L1p}}

See Section \ref{sec_prl1p}.

\section{MSE bounds for Thomson's multitaper}
\label{sec_conclusion}

In \cite[Section IV]{Thomson}, Thomson estimated $Bias(\widehat{S}_{(K)})$
by using the approximation $\frac{1}{K}\rho _{K}(N,W,\cdot )\approx
\frac{1}{2W}\mathbf{1}_{[-W,W]}$. Using Theorem \ref{th_main} we inspect that approximation:
\begin{equation*}
\left\vert Bias(\widehat{S}_{(K)}(\xi ))\right\vert \leq \left\vert \left|
S* \frac{1}{K}\rho_{K}(N,W,\cdot) - S*\frac{1}{2W}\mathbf{1}_{[-W,W]}
\right| \right\vert_\infty +\left\vert \left| S-S*\frac{1}{2W}\mathbf{1}_{[-W,W]} \right|
\right\vert_\infty
\end{equation*}
and, if $S$ is a bounded function, then Theorem \ref{th_main} implies that
\begin{equation*}
\left\vert \left| S*\frac{1}{K}\rho_{K}(N,W,\cdot)- S*\frac{1}{2W}\mathbf{1}
_{[-W,W]} \right| \right\vert_\infty \lesssim
\left\vert \left\vert S \right\vert\right\vert_\infty \frac{\log N}{K}.
\end{equation*}
(Similar considerations apply to the modified estimator where tapered
periodograms are weighted by the corresponding eigenvalue; in that case we
invoke Theorem \ref{L1p}.) The remaining term $\left\vert \left| S-S*\frac{1}{2W}\mathbf{1}_{[-W,W]}
\right| \right\vert_\infty $ can be bounded by
assuming that $S$ is smooth. For example, if, as in Thomson's work, $S$ is
assumed to $C^2$ (as a periodic function), then $\left\vert \left|
S-S*\frac{1}{2W}\mathbf{1}_{[-W,W]} \right| \right\vert_\infty \lesssim W^{2}$,
leading to the bias estimate:
\begin{equation}
\mathrm{Bias}(\widehat{S}_{(K)}(\xi ))\lesssim W^{2}+\frac{\log N}{K}.
\label{bias_estimate}
\end{equation}
On the other hand, for a slowly varying spectrum $S$, Thomson \cite{Thomson}
argues that
\begin{equation}  \label{eq_var}
\mathrm{Var}\left( \widehat{S}_{(K)}(\xi )\right) \lesssim \frac{1}{K}\text{, }
\end{equation}
(see, \cite{WMP}, \cite{liro08} or \cite[Section 3.1.2]{HogLak} for precise
expressions for the variance; in particular \cite[Theorem 2]{liro08} for the
bound in \eqref{eq_var}, valid when $S \in L^\infty$.) Given a number of
available observations, the estimates in \eqref{bias_estimate} and \eqref{eq_var} show how much bias
can be expected, in order to bring the
variance down by a factor of $1/K$. This leads to a concrete estimate for
the mean squared error
\begin{align*}
\mathrm{MSE}(\widehat{S}_{(K)}(\xi)) &= \mathbb{E}(S(\xi)-\widehat{S}_{(K)}(\xi))^2
\\ &=
\mathrm{Bias}(\widehat{S}_{(K)}(\xi))^2 + \mathrm{Var}(\widehat{S}_{(K)}(\xi))
\\
&\lesssim W^{4}+\frac{\log^2 N}{K^2} + \frac{1}{K},
\end{align*}
that can be used to decide on the value of the bandwidth resolution
parameter $W$.

Thus, in the slowly varying regime, the
error due to spectral leakage is largely dominated by the variance and
therefore, in agreement with Thomson's analysis, the mean squared error is $\lesssim W^{4} +
\frac{1}{K} \asymp \left(\frac{K}{N} \right)^{4} + \frac{1}{K}$. We note that the value of $K$ that
minimizes
this expression satisfies $K \asymp N^{4/5}$ and gives
\begin{align*}
\mathrm{MSE}(\widehat{S}_{(K)}(\xi)) \lesssim N^{-4/5}.
\end{align*}
A similar relation holds for
the so-called minimum bias sinusoidal tapers \cite{risi95}.
(Recall that $K$ and $W$ are related by $K= \lfloor 2NW \rfloor$; choosing
a certain value for $K$ amounts to choosing a corresponding
value for $W$.)

\section{Compressive acquisition of multi-band signals}
\label{sec_cs}
An important step in signal processing is to provide finite-dimensional models that
adequately capture analog phenomena. This problem is delicate and naive discretizations can lead to poor reconstruction.
A concrete instance of the modeling problem occurs in the study of the compressive acquisition of multi-band signals:
their naive representation suffers from the so-called DFT leakage. (See \cite{DB} and \cite{adha16} for the modeling problem in compressive sensing, including multi-scale settings.)

Let $\Delta \subseteq \mathbb{R}$ be an interval
with $\abs{\Delta} \leq 1$
that is decomposed
as union of disjointly translated copies of a smaller interval $[-W,W]$:
\begin{equation*}
\Delta =\bigcup_{j=0}^{M-1}[-W,W]+\{2Wj\}.
\end{equation*}
We consider a so-called multi-band signal $x(t)$ whose Fourier transform is
supported on the union of $L\ll M$ translated copies of $[-W,W]$, say
\begin{equation}
\label{eq_delst}
\Delta _{\ast }=\bigcup_{n=0}^{L-1}[-W,W]+\{2Wj_{n}\},
\end{equation}
with $j_{n}\in \{0,\ldots ,M-1\}$. In order to acquire such a signal, the
problem is to efficiently represent the corresponding sampled vector
\begin{equation*}
\mathbf{x}=(x(0),\ldots ,x(N-1)).
\end{equation*}
Davenport and Wakin \cite{DW} proposed using the dictionary of modulated
Slepian sequences
\begin{equation}
\label{eq_D}
\mathcal{D}:=\{e^{-2\pi i(2Wj)(\boldsymbol{\cdot })}v^{(k)}:k=0,\ldots
K-1,j=0,\ldots ,M-1\},
\end{equation}
where $K\approx 2NW$ is a parameter and $v^{(k)}=v^{(k)}(N,W)$ are the
discrete prolate spheroidal sequences. The sampled vector $\mathbf{x}$
is expected to have an
approximately $LK$-sparse representation in $\mathcal{D}$. More precisely,
the sampled vector $\mathbf{x}$ associated with a signal $x$ with Fourier
transform supported on the set $\Delta _{\ast }$ from \eqref{eq_delst} is
expected to be approximately captured by the sub-dictionary
\begin{equation}
\label{eq_dst}
\mathcal{D}_{\ast }:=\{e^{-2\pi i(2Wj_{n})(\boldsymbol{\cdot })}v^{(k)}:k=0,\ldots K-1,n=0,\ldots
,L-1\}.
\end{equation}
For the critical value $K=\left\lfloor 2NW\right\rfloor $ we obtain the
following average case estimate.

\begin{theorem}
\label{th_cs}
Let $x(t)=\tfrac{1}{L}\sum_{n=0}^{L-1} x_n(t)$ be the average of
independent, continuous-time, zero-mean, Gaussian stationary random
processes $x_n(t)$ with corresponding power spectra $S_{x_n} :=
\frac{1}{2W}\mathbf{1}_{[-W,W]+\{2Wj_n\}}$. Let $\mathcal{D}_*$ be the dictionary in \eqref{eq_dst}
with $K=\left\lfloor 2NW\right\rfloor$. For $N \geq 2$, let $\mathbf{x}=(x(0), \ldots, x(N-1))$ be
the vector of finite samples of $x$,
and $P_{\mathcal{D}_*} \mathbf{x}$ its projection onto the linear span of $\mathcal{D}_*$. Then
\begin{align*}
\frac {\mathbb{E} \left\{ || \mathbf{x}- P_{\mathcal{D}_*} \mathbf{x} ||^2
\right\}} {\mathbb{E} \left\{ || \mathbf{x} ||^2 \right\}} \lesssim \frac{L
\log N}{K}.
\end{align*}
\end{theorem}

\begin{proof}
Theorem 4.4 in \cite{DW} gives ${\mathbb{E}\left\{ ||\mathbf{x}||^{2}\right\} }=N$ and
\begin{equation*}
{\mathbb{E}\left\{ ||\mathbf{x}-P_{\mathcal{D}_{\ast }}\mathbf{x}||^{2}\right\} }\leq
\frac{L}{2W}\sum_{k=K}^{N-1}\lambda _{k}(N,W).
\end{equation*}
Since $2NW=\sum_{k=0}^{N-1}\lambda _{k}(N,W)$ and $K=\left\lfloor
2NW\right\rfloor $, we can use Lemma \ref{lemma_sum_eig} to conclude that
\begin{eqnarray*}
\frac{\mathbb{E}\left\{ ||\mathbf{x}-P_{\mathcal{D}_{\ast }}\mathbf{x}||^{2}\right\}
}{\mathbb{E}\left\{ ||\mathbf{x}||^{2}\right\} } &\leq &\frac{L}{K}\sum_{k=K}^{N-1}\lambda _{k}(N,W)
\\
&\leq &L\left( 1-\frac{1}{K}\sum_{k=0}^{K-1}\lambda _{k}(N,W)+\frac{1}{K}\right)  \\
&\lesssim &\frac{L\log N}{K}.
\end{eqnarray*}
\end{proof}

\begin{remark}
Approximation estimates such as the one in Theorem \ref{th_cs}
quantify the rate at which a finite dimensional model captures an analog phenomenon and are
therefore instrumental to the quantification of the so-called stable sampling rate in compressive
sensing \cite{adha16}.
\end{remark}

\begin{remark}
Theorem \ref{th_cs} applies to the critical number of Slepian sequences. Better approximation rates are possible by increasing the number of dictionary elements, see \cite[Section 5]{DW} and \cite[Corollary 3.12]{zw16}.
\end{remark}

\section{Conclusions and outlook}
We provided a quantitative description of the spectral window of Thomson's
multitaper method, leading to MSE bounds. We also quantified in the mean-squared sense the
effectiveness of the dictionary of modulated Slepian functions to capture analog multi-band signals.

An accumulation phenomenon similar to the one in Theorem \ref{th_main}
has been investigated in \cite{AGR} and numerically illustrated in \cite{BB, SimmonsGeoph2008}.
We therefore expect our approach to be applicable to other estimators including the one based on
spherical Slepians \cite{SIAMRev,SimmonsGeoph2008}, and those for non-stationary spectra
\cite{martin1985wigner,hlmakiko00,BB,OW,pfzh,waha07,omto16}.

Eigenvalue estimates for the spectral concentration problem
are available in the context of Hankel bandlimited functions \cite{AB}. We expect these
to be useful for problems involving $2D$ functions whose spectrum lies on
a disk, a setting relevant in cryo-electron microscopy, where estimation of noise stochastics is an
important consideration when applying PCA \cite{zhsi13}.

Dictionaries similar to the one in \eqref{eq_D} also appear in \cite{hola15},
in the context of prolate-spheroidal functions on the line. As shown in
\cite{hola15},
corresponding frame properties are related to a variant of Thomson's spectral window.
For this reason, it would be interesting to obtain a version of Theorem \ref{th_main} for the
$L^{\infty }$-norm.

\section{Technical lemmas}
\label{sec_tech}

\subsection{Trigonometric polynomials and Toeplitz operators}
Our proof uses tools from the Landau-Pollack-Slepian theory
\cite{Slepian,posl61,la67-1,la75-1,lapo61}.
For notational convenience, we use a temporal normalization that is slightly
different from the one in Section \ref{sec_tom} (this has no impact in the
announced estimates). We consider the space of trigonometric polynomials
\begin{equation*}
\mathcal{P}_{N}=Span\left\{ e_{\frac{-N+1}{2}+j}:0\leq j\leq N-1\right\}
\subseteq L^{2}\left( I\right) \text{.}
\end{equation*}
This is a Hilbert space with a reproducing kernel given by the translated
Dirichlet kernel, $\mathbf{D}_{N}(x-y)$, \ $x,y\in I$,\ $N \geq 2$, where $\mathbf{D}_{N}$ is given
by \eqref{Dirichlet}. Note that $\int_{I}\left\vert
\mathbf{D}_{N}\right\vert ^{2}=N$.

For $W\in (-1/2,1/2)$ the \emph{Toeplitz operator} $H_{W}^{N}$ is
\begin{equation}
H_{W}^{N}f:=P_{{\mathcal{P}_{N}}}\left( (P_{\mathcal{P}_{N}}f)\cdot
1_{[-W,W]}\right) ,\qquad f\in L^{2}(I),  \label{eq_toep_op}
\end{equation}
where $P_{{\mathcal{P}_{N}}}$ is the orthogonal projection onto $\mathcal{P}
_{N}$. When $f\in \mathcal{P}_{N}$, $H_{W}^{N}f$ is simply the projection of
$f\cdot 1_{[-W,W]}$ into $\mathcal{P}_{N}$. The Slepian functions $\{U_{k}(N,W):k=0,\ldots ,N-1\}$
are the eigenfunctions of $H_{W}^{N}$ with
corresponding eigenvalues $\lambda_{k}=\lambda_{k}(N,W)$:
\begin{equation}  \label{eq_eigen}
\int_{-W}^{W}\left\vert U_{k}(N,W;\xi )\right\vert ^{2}\,d\xi =\lambda _{k},
\end{equation}
ordered non-increasingly. We normalize the Slepian functions by $\int_{I}\left\vert U_{k}(N,W;\xi
)\right\vert ^{2}\,d\xi =1$.

\subsection{Integral kernels}

The Toeplitz operator $H_{W}^{N}$ from \eqref{eq_toep_op} can be explicitly
described by the formula
\begin{equation*}
H_{W}^{N}f(x)=\int_{I}f(y)K_{W}^{N}(x,y)dy\text{,}
\end{equation*}
where the kernel $K_{W}^{N}(x,y)$ is
\begin{equation}
K_{W}^{N}(x,y)=\int_{[-W,W]}\mathbf{D}_{N}(x-z)\overline{\mathbf{D}_{N}(y-z)}
dz\text{.}  \label{eq_kernel1}
\end{equation}

The operator $H_{W}^{N}$ can be diagonalized as:
\begin{equation}  \label{eq_diag}
H_{W}^{N}f=\sum_{k=0}^{N-1} \lambda _{k}(N,W)\left\langle
f,U_{k}(N,W)\right\rangle U_{k}(N,W)\text{.}
\end{equation}
The diagonalization in \eqref{eq_diag} means that the integral kernel in \eqref{eq_kernel1} can be
written as
\begin{align}  \label{eq_kernel2}
K_{W}^{N}(x,y)=\sum_{k=0}^{N-1} \lambda_{k}(N,W) U_{k}(N,W)(x)\overline{
U_{k}(N,W)(y)}, \qquad x,y \in I.
\end{align}

In particular, taking $x=y=\xi \in I$ yields,
\begin{equation}  \label{eq_aa}
\left( 1_{[-W,W]}\ast \left| \mathbf{D}_{N} \right| ^2\right) (\xi
)=\sum_{k= 0}^{N-1}\lambda_{k}(N,W)\left\vert U_{k}(N,W,\xi )\right\vert ^{2}\text{.}
\end{equation}

\subsection{An approximation lemma}

\begin{lemma}
\label{MainLemma} Let $f:I\rightarrow \mathbb{C}$ an integrable function, of
bounded variation, and supported on $I^{\circ }=(-1/2,1/2)$. For $N\geq 2$,
let
\begin{equation*}
f\ast \left\vert \mathbf{D}_{N}\right\vert ^{2}(x)=\int_{I}f(y)\left\vert
\mathbf{D}_{N}\left( x-y\right) \right\vert ^{2}dy,\qquad x\in I.
\end{equation*}
Then
\begin{equation}
\left\Vert f-\frac{1}{N}f\ast \left\vert \mathbf{D}_{N}\right\vert
^{2}\right\Vert _{L^{1}(I)}\lesssim Var\left( f,I\right) \frac{\log N}{N}.
\end{equation}
\end{lemma}

\begin{remark}
In the above estimate, $Var(f,I)$ denotes the total variation of $f$ on $I$.
If $f=1_{[-W,W]}$, with $W\in (-1/2,1/2)$, then $Var\left( f,I\right) =2$
and the estimate reads
\begin{equation*}
\left\Vert \mathbf{1}_{[-W,W]}-\frac{1}{N}\mathbf{1}_{[-W,W]}\ast \left\vert
\mathbf{D}_{N}\right\vert ^{2}\right\Vert _{L^{1}(I)}\lesssim \frac{\log N}{N }.
\end{equation*}
\end{remark}

\begin{proof}
By an approximation argument, we assume without loss of generality that $f$
is smooth (see for example \cite[Lemma 3.2]{AGR}). We also extend $f$
periodically to $\mathbb{R}$. Note that this extension is still smooth
because $f|I$ is supported on $I^{\circ }$.

\textbf{Step 1}. Since $f(x+h)-f(x)=\int_{0}^{1}f^{\prime }(th+x)h\,dt$, we
can use the periodicity of $f$ to estimate
\begin{align*}
\left\Vert f(\cdot +h)-f\right\Vert _{L^{1}(I)}& \leq
\int_{0}^{1}\int_{-1/2}^{1/2}\left\vert f^{\prime }(th+x)\right\vert
dx\left\vert h\right\vert dt \\
& =\int_{0}^{1}\int_{-1/2+th}^{1/2+th}\left\vert f^{\prime }(x)\right\vert
dx\left\vert h\right\vert dt \\
& =\int_{0}^{1}\int_{-1/2}^{1/2}\left\vert f^{\prime }(x)\right\vert
dx\left\vert h\right\vert dt=Var(f,I)\left\vert h\right\vert .
\end{align*}
Since $f$ is periodic, the previous estimate can be improved to:
\begin{equation}
\left\Vert f(\cdot +h)-f\right\Vert _{L^{1}(I)}\lesssim Var(f,I)\left\vert
\sin (\pi h)\right\vert ,\qquad h\in \mathbb{R}.  \label{eq_f}
\end{equation}

\textbf{Step 2}. We use the notation $f^{N}:=f\ast \tfrac{1}{N}\left\vert
\mathbf{D}_{N}\right\vert ^{2}$. By a change of variables and periodicity,
\begin{equation*}
f(x)-f^{N}(x)=\frac{1}{N}\int_{-1/2}^{1/2}(f(x)-f(y+x))\left\vert \mathbf{D}
_{N}(-y)\right\vert ^{2}dy.
\end{equation*}
We can now finish the proof by resorting to \eqref{eq_f}:
\begin{eqnarray*}
\left\Vert f-f^N\right\Vert _{L^{1}(I)} &\lesssim &Var(f,I) \frac{1}{N}
\int_{-1/2}^{1/2}\left\vert \sin (\pi y)\right\vert \left\vert \mathbf{D}
_{N}(y)\right\vert ^{2}dy \\
&\lesssim &Var(f,I)\frac{1}{N}\int_{0}^{1/2}\frac{\left\vert \sin (\pi
Ny)\right\vert }{\left\vert y\right\vert }dy \\
&\lesssim &Var(f,I)\frac{1}{N}\left[ 1+\int_{1}^{N/2}\frac{1}{\left\vert
y\right\vert }dy\right] \\
&\lesssim &Var(f,I)\frac{\log N}{N}\text{.}
\end{eqnarray*}
\end{proof}

\subsection{Proof of Theorem {\protect\ref{L1p}}}

\label{sec_prl1p} The result follows immediately from Lemma \ref{MainLemma}
and \eqref{eq_aa}.

\subsection{Proof of Lemma \protect\ref{lemma_sum_eig}}
The main idea of this proof is to directly estimate the average of the critical number of
eigenvalues instead of building on individual estimates. This solves a problem posed in
\cite[Remarks after Theorem 5]{liro08}.

We first note from \eqref{eq_kernel1} that
\begin{equation}
\mathrm{trace}\left( H_{W}^{N}\right)
=\int_{I}K_{W}^{N}(x,x)dx=\int_{[-W,W]}\int_{I}\left\vert \mathbf{D}
_{N}(x-y)\right\vert ^{2}dydx=2NW\text{,}  \label{eq_trace}
\end{equation}
since\ $\int_{I}$\ $\left\vert \mathbf{D}_{N}\right\vert ^{2}=N$. Moreover a
similar calculation gives
\begin{equation*}
\mathrm{trace}\left( H_{W}^{N}\right) ^{2}=\int_{[-W,W]}\int_{I}\mathbf{1}
_{_{[-W,W]}}(y)\left\vert \mathbf{D}_{N}(x-y)\right\vert ^{2}dydx\text{.}
\end{equation*}
Hence we can use Lemma \ref{MainLemma} to conclude that
\begin{align*}
& \mathrm{trace}\left[ \left( H_{W}^{N}\right) -\left( H_{W}^{N}\right) ^{2}
\right] =\int_{-W}^{W}\left[ N\mathbf{1}_{_{[-W,W]}}(x)-\left( \mathbf{1}
_{_{[-W,W]}}\ast \left\vert \mathbf{D}_{N}\right\vert ^{2}\right) (x)\right]
dx \\
\qquad & \leq \int_{I}\left\vert N\mathbf{1}_{_{[-W,W]}}(x)-\left( \mathbf{1}
_{_{[-W,W]}}\ast \left\vert \mathbf{D}_{N}\right\vert ^{2}\right)
(x)\right\vert dx \\
& \leq C\log N,
\end{align*}
for some constant $C$. Using this bound, we estimate:
\begin{align*}
C\log N& \geq \sum_{k=0}^{N-1}\lambda _{k}(1-\lambda
_{k})=\sum_{k=0}^{K-1}\lambda _{k}(1-\lambda _{k})+\sum_{k=K}^{N-1}\lambda
_{k}(1-\lambda _{k}) \\
& \geq \lambda _{K-1}\sum_{k=0}^{K-1}(1-\lambda _{k})+(1-\lambda
_{K-1})\sum_{k=K}^{N-1}\lambda _{k} \\
& =\lambda _{K-1}K-\lambda _{K-1}\sum_{k=0}^{K-1}\lambda _{k}+(1-\lambda
_{K-1})(2NW-\sum_{k=0}^{K-1}\lambda _{k}) \\
& =\lambda _{K-1}K+2NW(1-\lambda _{K-1})-\sum_{k=0}^{K-1}\lambda _{k} \\
& =2NW-\sum_{k=0}^{K-1}\lambda _{k}+\lambda _{K-1}(K-2NW) \\
& \geq K-\sum_{k=0}^{K-1}\lambda _{k}-1.
\end{align*}
Hence,
\begin{equation*}
K-\sum_{k=0}^{K-1}\lambda _{k}\leq C\log N+1\text{.}
\end{equation*}
On the other hand $\sum_{k=0}^{K-1}\lambda _{k}-K\leq 2NW-K\leq 1$.
Therefore, since $N\geq 2$,
\begin{equation*}
\left\vert K-\sum_{k=0}^{K-1}\lambda _{k}\right\vert \lesssim \log N\text{,}
\end{equation*}
and the conclusion follows.

\section*{Acknowledgement}
The authors are very grateful to Radu Balan,
Franz Hlawatsch,
Christian Kase{\ss}, \mbox{Wolfgang} Kreuzer, 
Jo\~ao M. Pereira and Frederik Simons for insightful conversations that helped improve the 
article.

\end{document}